\documentclass[onecolumn, twoside]{IEEEtran}
\IEEEoverridecommandlockouts

\usepackage{ifpdf}
\usepackage{cite}
\usepackage[cmex10]{amsmath}
\usepackage{enumerate}
\usepackage{amsthm}
\usepackage{latexsym}
\usepackage{amssymb}
\usepackage{xspace}
\usepackage{amscd}
\usepackage{enumerate}
\usepackage{graphicx}
\usepackage{epic}

\newcommand{\N}{{\mathbb N}}

\newcommand{\Z}{{\mathbb Z}}
\newcommand{\F}{\mathbb F}

\newcommand{\C}{{\mathcal C}}

\newcommand{\tq}{\; \mid \;}

\newcommand{\PAut}{\rm PAut}

\newcommand{\supp}{\rm supp}

\usepackage{algorithmic}
\usepackage{array}

\usepackage{mdwmath}
\usepackage{mdwtab}

\usepackage{stfloats}
\usepackage{color}

\newtheorem{theorem}{Theorem}
\newtheorem{lemma}[theorem]{Lemma}

\newtheorem{definition}[theorem]{Definition}

\newtheorem{remark}[theorem]{Remark}
\newtheorem{remarks}[theorem]{Remarks}

\begin{document}

\title{New advances in permutation decoding of first-order Reed-Muller codes  
\thanks{Supported by the Spanish Government under Grant PID2020-113206GBI00
funded by MCIN/AEI/10.13039/501100011033.
}}

\author{\IEEEauthorblockN{Jos\'e Joaqu\'{i}n Bernal
 and
Juan Jacobo Sim\'on (Member, IEEE). \\
\IEEEauthorblockA{Departamento de Matem\'{a}ticas\\
Universidad de Murcia,
30100 Murcia, Spain.\\ Email: \{josejoaquin.bernal, jsimon\}@um.es} \\
}}

\maketitle

\begin{abstract}
In this paper we describe a variation of the classical permutation decoding algorithm that can be applied to any affine-invariant code with respect to certain type of information sets. In particular, we can apply it to the family of first-order Reed-Muller codes with respect to the information sets introduced in \cite{BS}. Using this algortihm we improve considerably the number of errors we can correct in comparison with the known results in this topic. 
\end{abstract}

\IEEEpeerreviewmaketitle

\section{Introduction}

Permutation decoding was introduced by F. J. MacWilliams in \cite{macwilliams} and it is fully described  in \cite{handbook2} and \cite{MacSlo}. With a fixed information set for a given linear code, this technique uses a special set of its permutation automorphisms called PD-set. Then, the idea of permutation decoding is to apply the elements of the  PD-set to the received vector until the errors are moved out of the fixed information set. 

The existence of PD-sets relies on the information set previously taken as reference. In fact, it may happen that for a fixed error correcting code the election of the information set causes the existence or not of a PD-set. Many authors have studied families of codes for which it is possible to develop methods to find PD-sets with respect to certain types of information sets. It has been seen that all methods rely heavily on those families chosen. In this paper, our interest is the family of Reed-Muller codes.

The family of Reed-Muller codes was introduced by D. E. Muller in 1954 \cite{Muller} and a specific decoding algorithm for them was presented by I. S. Reed in the same year \cite{Reed}. Since then, many authors have paid attention to this family of codes for many reasons. On the one hand, they can be implemented and decoded easily, and on the other hand they have a rich algebraic structure which allows us to see them from different points of view. In particular, the family of Reed-Muller codes is contained in the family of so-called affine-invariant codes and so a defining set can be defined for them (see, for instance, \cite{handbook2}). 

The problem of applying permutation decoding to Reed-Muller codes has been addressed for many authors earlier, see for instance \cite{KMM, KrollVic, Senev} or \cite{BarrVill}. In \cite{KMM}, the authors give a construction of information sets for generalized Reed-Muller codes. Then, in \cite{Senev} the author uses those information sets to apply the permutation decoding algorithm to first-order Reed-Muller codes; specifically, it is proven that one can find PD-sets for $R(1,m)$, with $m>4$, in order to correct up to $4$ errors. Later, in \cite{KMM2}, the authors extend the method given by P. Seneviratne to find $(m-1)$-PD-sets, for $m\geq 5$, and $(m+1)$-PD-sets for $m\geq 6$. Moreover, they show that the set of translations of $\F_{2}^m$ is an $s$-PD-set for $s=\min\left\{\left\lfloor\frac{2^m-1}{1+m} \right\rfloor, 2^{m-2}-1\right\}$. Finally, in \cite{KMM3}, the same authors, find $(M-1)$-PD-sets for any generalized Reed-Muller code $R(r,n)$, where $1\leq r\leq \frac{n-1}{2}$, provided that a $(n,M,2r+1)$ binary code exists. Furthermore, they construct $s$-PD-sets for $R(1,m)$ of size $s+1$, contained in the set of translations of $\F_2^m$ (this is the smallest size according to the Gordon-Schönheim bound (\cite{Gordon, Scho})) using the same ideas that were introduced in \cite{BarrVill}.
 
In this paper we define a variation of the classical permutation decoding algorithm that it is valid for any affine-invariant code with respect to a specific kind of information sets. In particular, it can be used successfully for first-order Reed-Muller codes, $R(1,m)$, with respect to the information sets introduced in \cite{BS}. Then, under some simple conditions, we show that this algorithm improves notably the number of errors we can correct in comparison with the known results (see Section \ref{Examples}).  

First we include a section devoted to describe the preliminaries, then we present the classical permutation decoding in the context of affine-invariant codes. In Section \ref{infosets} we recall the construction of information sets for first-order Reed-Muller codes given in \cite{BS}. In Sections \ref{NewPD} and \ref{PDRM} we include our main results. We introduce the new permutation-decoding algorithm
and we show how it can be applied to first-order Reed-Muller codes, respectively. Finally, in Section \ref{Examples}, we present some relevant examples that improve the number of errors that can be corrected.

\section{Preliminaries}\label{Preliminaries}

  In this paper we see the family of Reed-Muller codes as a subfamily of the so-called affine-invariant codes (see, for instance,\cite{handbook2},\cite{Hb2} or \cite{Charpin} ). To introduce this kind of codes we work in the ambient space $\F G$, the group algebra with $\F$ the binary field and $G$ the additive subgroup of the field with $2^m$ elements, where $m\in \N$ and $n=2^m-1$ are fixed all throughout this paper. Note that $G$ is an elementary abelian group of order $|G|=2^m$ and $G^*=G\setminus\{0\}$ is a cyclic group.
	
	From now on we also fix $\alpha\in G^*$ a generator element, that is, $\langle \alpha \rangle=G^*$ (note that $\alpha$ is a primitive $n$-th root of unity). Then we write the elements in $\F G$ as
	$$b X^0+\sum\limits_{i=0}^{n-1} a_i X^{\alpha^i}\quad (a_i,b\in \F)$$
	In this work a code is an ideal in $\F G$; we write $\C\subseteq \F G$.	
	
	To introduce the family of affine-invariant codes first we need to give the notion of extended cyclic code in $\F G$.
	
	\begin{definition}
 A code $\C\subseteq \F G$ is an extended cyclic code if for any $b X^0+\sum\limits_{i=0}^{n-1} a_i X^{\alpha^i}\in \C$ one has that $b X^0+\sum\limits_{i=0}^{n-1} a_i X^{\alpha^{(i+1)}}\in \C$ and $b +\sum\limits_{i=0}^{n-1} a_i =0$.
\end{definition}

For any extended cyclic code $\C\subseteq \F G$ we denote by $\C^*$ the punctured code at the position $X^0$. Then $\C^*$ is a \textit{cyclic} code in the sense that	it is the projection to $\F G^*$ of the image of a cyclic code via the map 
 \begin{eqnarray}\label{inmersionciclicos}
  \nonumber \F[X]/\langle X^n-1\rangle &\longrightarrow& \F G\\
   \sum\limits_{i=0}^{n-1}a_{i}X^i&\hookrightarrow& \left(-\sum\limits_{i=0}^{n-1} a_i\right)X^0+\sum\limits_{i=0}^{n-1}a_i X^{\alpha^i},	  
  \end{eqnarray}
	
where $\alpha$ is the fixed $n$-th root of unity and $\F[X]/\langle X^n-1\rangle$ is the quotient algebra of $\F[X]$, the polynomials with coefficients in $\F$.  
	
	On the other hand, we denote by $S_G$ the group of automorphisms of $G$. Then, we see $S_G$ acting on $\F G$ via 
		$$\tau\left(b X^0+\sum\limits_{i=0}^{n-1} a_i X^{\alpha^i}\right)=b X^{\tau(0)}+\sum\limits_{i=0}^{n-1} a_i X^{\tau(\alpha^i)}$$
	for any $\tau \in S_G$. 
	
	In this context, we define 
	$$\PAut(\C)=\{\tau\in S_{G}\mid \tau(\C)=\C\}$$
	
	Analogously, in the case of the cyclic punctured code $\C^*$, we consider $S_{G^*}$ the group of automorphisms of $G^*$, we see it acting on $\F G^*$, and we define $\PAut(\C^*)$. So, we can identify any $\tau\in S_{G^*}$ with the corresponding automorphism in $S_G$ that fixes the position $X^0$ and we may write $\PAut(\C^*)\subseteq \PAut(\C)$.

	Now, we may introduce the family of affine-invariant codes.	

\begin{definition}
	 We say that a code $\C\subseteq \F G$ is an affine-invariant code if it is an extended cyclic code and ${\rm GA}(G,G)\subseteq \PAut(\C)$ where
	$${\rm GA}(G,G)=\{x\mapsto ax+b\mid x\in\F G, a\in G^*, b\in G\}.$$
	\end{definition}

	Once we have introduced the ambient space and the general family of affine-invariant codes, we deal with the definition of Reed-Muller codes. 
	
For any $s\in\{0,\dots,n=2^m-1\}$ we consider the $\F$-linear map $\phi_s:\F G\rightarrow G$ given by
\begin{equation*}
 \phi_s\left(b X^0+\sum\limits_{i=0}^{n-1} a_i X^{\alpha^i}\right)=0^s+\sum\limits_{i=0}^{n-1} a_i \alpha^{is}
\end{equation*}
where we assume $0^0=1\in\F$ by convention.

\begin{definition}
 Let $\C\subseteq\F G$ be an affine-invariant code. The set 
 $$D(\C)=\{i\mid \phi_i(x)=0 \text{ for all } x\in \C\}$$
 is called the defining set of $\C$.
\end{definition}

It is well known that any affine-invariant code is totally determined by its defining set (see, for instance,\cite{handbook2}).

Finally, we also need to recall the notions of binary expansion and $2$-weight. 

\begin{definition}
For any natural number $k$ its binary expansion is the sum $\sum_{r\geq 0} k_r 2^r=k$ with $k_r\in \{0,1\}$. The $2$-weight or simply weight of $k$ is ${\rm wt}(k)=\sum_{r\geq 0}k_r$. 
\end{definition}

\begin{definition}\label{defRM}
 Let $0< \rho\leq m$. The Reed-Muller code of order $\rho$ and length $2^m$, denoted by $R(\rho,m)$, is the affine-invariant code in $\F G$ with defining set 
 $$D(R(\rho,m))=\{i\mid 0\leq i<2^m-1 \text{ \rm{and} } {\rm wt}(i)< m-\rho\}.$$
\end{definition}

In this paper we deal with first-order Reed-Muller codes, that is, the codes with defining set
$$D(R(1,m))=\{i\mid 0\leq i<2^m-1 \text{ \rm{and} } {\rm wt}(i)< m-1\}.$$
We denote by $R^*(1,m)$ the cyclic punctured code at the position $X^0$.

\begin{remarks}
 Note that some special values of $m$ give us trivial cases in which we do not have interest in the context of this paper. Specifically
  \begin{itemize}
	 \item $m=1=\rho$ gives the code $R(1,1)=\F G$, the entire group algebra.
	 \item $m=2$ gives us the code $R(1,2)$ which is the code of all even weight codewords in $\F G$.
 \end{itemize}
So, from now on we consider that $m>2$.

\end{remarks}

It is proven that the code $R(r,m)$ is a binary $[2^m,\sum\limits_{i=0}^r{m\choose i},2^{m-r}]$ code (see, for instance, \cite{handbook2}); so, the code $R(1,m)$ is a binary $[2^m,m+1,2^{m-1}]$ code. 
	
\section{Information sets from defining sets for first-order Reed-Muller codes}\label{infosets}

To begin this section let us establish the definition of information set in the ambient space $\F G$.

\begin{definition}
 A set $I\subseteq \{0,\alpha^0,\dots,\alpha^{n-1}\}=G$ is an information set for a code $\C\subseteq \F G$, with dimension $k$, if $|I|=k$ and the set of the projections of the codewords of $\C$ onto the positions in $I$ is equal to $\F^{|I|}$.
\end{definition}

For any affine-invariant code we may also define the notion of information set for $\C^*$ as a set contained in $G^*=\{\alpha^0,\dots, \alpha^{n-1}\}$. Clearly, an information set for $\C^*$ is an information set for $\C$, but if $I\subseteq G$ is an information set for $\C$ and $0\in I$ then $I'=I\setminus\{0\}$ is not an information set for $\C^*$.

In \cite{BS} we showed how to obtain an information set for $R(1,m)$ only in terms of its parameters $m$ and $n=2^m-1$. To achieve those results we need to impose the following restrictions:
\begin{equation}\label{condFirstOrder}
 n=r_1\cdot r_2,\qquad \gcd(r_1,r_2)=1 \qquad \text{ and } \qquad  r_1, r_2>1
\end{equation}

Then, from now on, we fix an arbitrary isomorphism $\varphi:\Z_n\longrightarrow \Z_{r_1}\times\Z_{r_2}$.

\begin{theorem}[\cite{BS}] \label{teoremainfosetprimerorden}
  Under the restrictions given above we have that the set $\{0,\alpha^i\mid i\in \varphi^{-1}\left(\Gamma\right)\}$ where 
 $$\Gamma=\Gamma(\C)=\left\{(i_1,i_2)\in\Z_{r_1}\times\Z_{r_2} \tq 0\leq i_1< a, 0\leq i_2<\frac{m}{a}\right\},$$
 is an information set for $R(1,m)$, with $a$ the order of $2$ modulo $r_1$ ($Ord_{r_1}(2)$). 
\end{theorem} 

Note that, as we have seen before, the set $I'=I\setminus\{0\}$ is not an information set for $R^*(1,m)$.

\medskip

To finish this section, we include Table I which shows that conditions (\ref{condFirstOrder}) are not too restrictive, specifically we include the suitable values of $m$ up to length 2048. The values $m=3,5,7$ yield a prime number for $n=2^m-1$. The columns $\mathbf{t}$ and $\mathbf{a}$ represent the error-correction capability and $Ord_{r_1}(2)$ respectively.

\begin{table}[h]
\begin{center}
\caption{Parameters for first order RM codes up to length 2048}
\end{center}

\begin{center}
\begin{tabular}{|c|c|c|c|c|}
 \hline
 $\mathbf{m}$&$\mathbf{n}$&$\mathbf{r_1}$&$\mathbf{r_2}$& $\mathbf{a}$\\ \hline
 4&5&3&5&2\\ \hline
 6&63&7&9&3\\ \hline
 8&255&3&85&2\\ \hline
 8&255&15&17&4\\ \hline
 8&255&5&51&4\\ \hline
 9&511&7&73&3\\ \hline
 10&1023&3&341&2\\ \hline
 10&1023&11&93&10\\ \hline
 10&1023&31&33&5\\ \hline
 11&2047&23&89&11\\ \hline
\end{tabular}
\end{center}

\end{table}

\section{Permutation decoding for affine-invariant codes}\label{NewPD}

  \subsection{The Permutation decoding algorithm in $\F G$}
	
	Let $\C\subseteq \F G$ be a $t$-error correcting code with dimension $k$ and let $I\subseteq G$ be an information set for $\C$. We say that a matrix $G_{k\times n}$ with coefficients in $\F$ is a generator matrix whether its rows form a basis for $\C$ as $\F$-vector space. As usual a parity check matrix is a generator matrix for the dual code of $\C$ (defined as usual). Moreover, we say that a parity check matrix is in standard form with respect to $I$ if the columns corresponding to the positions not in $I$ are the columns of the identity matrix of order $n-k$.
	
	As we have said in the introduction, the permutation decoding algorithm is based on the existence of certain special subsets in $\PAut(\C)$ called PD-sets. Let us give that definition.
	
	\begin{definition}
	 We say that $P\subseteq \PAut(\C)$ is an $s$-PD-set ($s\leq t$) with respect to $I$ if any $s$ positions in $G=\{0,1,\alpha,\dots,\alpha^{n-1}\}$ are moved out of $I$ by at least one element in $P$. 
	
	In case $s=t$ we say that $P$ is a PD-set. In case $s<t$ we talk about a partial PD-set.
	
	\end{definition}
	
	Now, let $c\in \C$ be the message and $r=c+e$ the received word, where $e$ represents the error vector. For any $v\in \F G$ we denote its weight as $\omega(v)$, that is, the number of coefficients (as an element in $\F G$) different from zero. We assume that $\omega(e)\leq t$. The symbols of $r$ in the positions of $I$ are called its information symbols.

	The following well known result is essential to define the algorithm.
	
	\begin{theorem}(\cite{MacSlo})\label{PD}
	 Let $H$ be a parity check matrix for $\C$ in standard form with respect to an information set $I$. Then, the information symbols of $r$ are correct if and only if $\omega(H\cdot r^T)\leq t$, the error correction capability of $\C$.
	\end{theorem}
	
	Then, the permutation decoding algorithm works as follows.
 
 \vspace{.4cm}

 \begin{quote}	
 {\scshape Algorithm I:}
	
	\begin{enumerate}
		\item Fix $I$ an information set for $\C$, a parity check matrix $H$ in standard form with respect to $I$ and $P\subseteq \PAut(\C)$ an $s$-PD-set ($s\leq t$) for $\C$ with respect to $I$.
		\item For each $\tau\in P$ we compute $\tau(r)$ and $H\cdot \tau(r)^T$ until we obtain $\tau_0$ such that $\omega(H\cdot\tau_0(r)^T)\leq t$. (From the definition of $s$-PD-set such an automorphism $\tau_0$ must exist.)
		\item We recover $c'\in \C$ from the information symbols of $\tau_0(r)$.
		\item We decode to $\tau^{-1}(c')=c$.
	\end{enumerate}

 \end{quote}

\medskip

By definition of partial $s$-PD-set and Theorem \ref{PD} we see that we will always be able to correct up to $s$ errors by using {\scshape Algorithm I}.
	
	\subsection{Permutation decoding for affine-invariant codes}
	
	In this section we present a modification of the classical PD-algorithm ({\scshape Algorithm I}) valid for any affine-invariant code. First, we introduce another type of subsets in $\PAut(\C)$ similar to the previous $s$-PD-sets.
	
	As above let $\C\subseteq \F G$ be an affine-invariant code and let $\C^*$ be the corresponding punctured cyclic code.
	
\begin{definition}
 Let $J\subseteq G^*=\{\alpha^i\mid i=0,\dots,n-1\}$ be an arbitrary set. 
 A set $P\subseteq \PAut(\C^*)\subseteq S_{G^*}$ is a $s$-PD-like set for  $\C$ with respect to $J$ if any $s$ positions in $G^*$ are moved out of $J$ by at least one element in $P$. In case $s=t$ we say that $P$ is a PD-like set. 
\end{definition}

Observe that if the set $J$ is an information set for $\C^*$ we are talking about PD-sets of $\C^*$ and consequently PD-sets of $\C$. (As we have noted before we can see $\PAut(\C^*)\subseteq \PAut(\C)$.)
\medskip

Now, let $I\subseteq \{0,\alpha^0,\dots,\alpha^{n-1}\}$ be an information set for $\C$. Let $P\subseteq \PAut(\C^*)\subseteq \PAut(\C)$ a $s$-PD-like set for $\C$ with respect to the set 
\begin{equation}\label{Iprima}
I'=\left\{
\begin{array}{lcl}
 I & \text{  if  } & 0\notin I\\
 I\setminus\{0\} &  & \text{otherwise}
\end{array}
\right.
\end{equation} 
(Note that if $0\in I$ then $I'$ is not an information set for $\C^*$). Then, if $0\notin I$ then $P$ actually acts as an $s$-PD-set for $\C$ with respect to $I$ so the algorithm is exactly the same as it was described in the previous section.

We also need to define the following automorphisms in $\PAut(\C)$.

\begin{definition}
  For any $k=0,1,\dots,n-1$ let $\sigma_k$ the automorphism of $G$ given by 
	
	$$\sigma_k\left(b X^0+\sum\limits_{i=0}^{n-1}a_i X^{\alpha^i}\right)=b X^{\alpha^k}+\sum\limits_{i=0}^{n-1}a_i X^{(\alpha^i+\alpha^k)} $$
\end{definition}

By definition of affine-invariant code it is clear that $\sigma_k\in \PAut(\C)$ for any $k=0,\dots, n-1$. We denote by $\Sigma$ the ordered set $\{1_G,\sigma_0,...\sigma_{n-1}\}$, where $1_G$ represents the identity automorphism in $\F G$. 

As above, let $r=c+e$ be the received word with $c\in \C$ and $e$ the error vector, where we assume $\omega(e)\leq t$. Then, the new algorithm works as follows:

\begin{quote}	
 {\scshape Algorithm II:}
	
	\begin{enumerate}
		\item Let I be an information set, H a parity check matrix in
standard form with respect to I, and $P\subseteq \PAut(\C^*)\subseteq\PAut(\C)$ a $s$-PD-like set for $\C^*$ ($s\leq t$) with respect to $I'$, defined in (\ref{Iprima}).

	\item We take $r=1_G(r)$. For each $\tau\in P$ we compute $\tau(r)$. If we find $\tau_0\in P$ such that $\omega\left(H\cdot \tau_0(r)^T\right)\leq t$ then the symbols of $\tau_0(r)$ are correct.
	
	In case there isn't any $\tau_0\in P$ satisfying the desired condition we go to Step 3, otherwise we go to Step 4.
	
	\item We compute $\sigma_i(r)$, where $\sigma_i$ is the next element in $\Sigma$, and we repeat Step 2 starting with $\sigma_i(r)$ instead of $r$.
	
		\item We recover $c'\in \C$ from the information symbols of $\tau_0(r)$.
		\item We decode to $\tau_0^{-1}(c')=c$.
	
\end{enumerate}

\end{quote}

Then, we have the following result:

\begin{theorem}\label{PDnew}
 Let $\C\subseteq \F G$ be an affine-invariant code with correction capability $t$. Let $I\subseteq \{0,\alpha^0,\dots,\alpha^{n-1}\}$ be an information set for $\C$. Let $P\subseteq \PAut(\C^*)\subseteq\PAut(\C)$ be a $s$-PD-like set for $\C^*$ with respect to $I'$ where $s\leq t$. Then we can correct up to $s$ errors by using Algorithm II.
\end{theorem}

\begin{proof}First, if $0\notin I (I'=I)$ then the statement is proven by the classical result for the permutation decoding algorithm because $P$ is actually an $s$-PD-set. In fact, we only need to use the element $1_G\in \Sigma$.

If $0\in I$ the existence of $\tau_0$ is given by the definition of $s$-PD-like set and by the assumption $\omega(e)\leq s\leq t$. Observe that since $\omega(e)\leq s$ we need to repeat Step 3) at most $s$ times in order to get the condition $0\notin \supp(e)$, so we decode at most in $s$ phases.

This finishes the proof.
\end{proof}

	By the previous theorem we always may correct up to $s$ errors in an affine-invariant code with $s\leq t$ provided that we find a $s$-PD-like set in $\C^*$ with respect to the set $I'$ associated to an information set $I$. 
	
	Therefore, the success of the new algorithm for a family of codes depends on whether we are able to describe suitable information sets in order to find PD-like sets with respect to them. In the following section, we are exhibiting the other main results of this paper, namely, we will prove that we can apply the algorithm for first-order Reed-Muller codes by using the information sets given in Section \ref{infosets}. 
	
	\begin{remark}[\textbf{Complexity}]
This new algorithm only differs from the original one in the inclusion of step $3$, that is, we repeat the action of the $s$-PD-like set several times. Given a code of length $n$ and correction capability $t\geq s$, although the set $\Sigma$ has $n$-elements in the worst-case we need to use exactly $s$ of them in order to move the support of the error vector out of position $0$. Therefore, since the worst-case time complexity for the classical permutation decoding algorithm is $\mathcal{O}(nkz)$, where $k$ is the dimension of the code and $z$ is the cardinality of the $s$-PD-set (see, for instance \cite{KMM3}), for the new algorithm the worst-case time complexity is $\mathcal{O}(snkz)$. 
	\end{remark}

\section{Permutation decoding for First-order Reed-Muller codes}\label{PDRM}

Let $\C=R(1,m)$ and suppose that $n=2^m-1$ satisfies the conditions (\ref{condFirstOrder}). Recall that we fixed an isomorphism $\varphi: \Z_n\longrightarrow \Z_{r_1}\times\Z_{r_2}$ and we consider the information set for $\C$ given in Section \ref{infosets}, that is
$$I=\{0,\alpha^i\mid i\in \varphi^{-1}(\Gamma)\}$$
where $\Gamma=\{(i_1,i_2)\in \Z_{r_1}\times\Z_{r_2}\mid 0\leq i_1\leq a, 0\leq i_2<\frac{m}{a}\}$, $a=Ord_{r_1}(2)$.

This set depends on the parameters $m, a$ which in turn depend on the decomposition $n=r_1\cdot r_2$. We are interested in select a decomposition which gives us a suitable information set to get our purposes, that is, to apply the new permutation decoding algorithm efficiently. To make the mentioned selection of the parameters we need the following technical lemmas. 

\begin{lemma}\label{descomposicion1}
 Let $m,\delta\in \N$. Then $m|\delta$ if and only if $2^m-1|2^\delta-1$.  
\end{lemma}
\begin{proof}
  First, assume that $m|\delta$. Then, by using the well-known formula
	$$X^k-1=(X-1)\cdot \sum\limits_{i=1}^{k}X^{k-i}$$
	with $X=2^m$ and $k=\delta/m$ we obtain that $2^m-1|2^\delta-1$ as we wanted.
	
	Now, assume that $2^m-1|2^\delta-1$. Then, $Ord_{2^m-1}(2)=m$ and $2^\delta\equiv 1 (\text{mod } 2^m-1)$ so $m|\delta$.
\end{proof}

\begin{lemma}\label{descomposicion2}
 Suppose that $2^m-1=r_1\cdot r_2$ with $\gcd(r_1,r_2)=1, r_1,r_2>1$. Then either $Ord_{r_1}(2)=m$ or $Ord_{r_2}(2)=m$.
\end{lemma}

\begin{proof}
 Let us denote $Ord_{r_1}(2)=a$ and $Ord_{r_2}(2)=b$.
 Firstly, since $2^m\equiv 1 (\text{mod } r_i)$ for $i=1,2$ we have that $a,b|m$. Then $\mu=lcm(a,b)|m$. On the other hand, $2^\mu\equiv 1 (\text{mod } r_i), i=1,2$ implies $r_1,r_2| 2^\mu-1$, and since $r_1,r_2$ are coprime one has $2^m-1=r_1\cdot r_2|2^\mu-1$. By the previous lemma we have that $m\leq \mu$, and so $m=\mu$.

Now, assume w.l.o.g. that $a\leq b$. In case $a=b$ we are done because this implies $a=b=\mu=m$. So, let us suppose that $a<b$; we are proving that $b=m$.

Since $r_1|2^a-1$ and $r_2|2^b-1$ then $r_1\cdot r_2|(2^a-1)\cdot (2^b-1)$ which yields
$$2^m-1=r_1\cdot r_2\leq (2^a-1)\cdot (2^b-1)=2^{a+b}+1-(2^a+2^b)<2^{a+b}-1,$$
note that $a, b\neq 0$ by hypothesis and then $2^a+2^b>2$. So, we have that $m<a+b$.

Let $d=gcd(a,b)$. It is well known that $m=\frac{ab}{d}$. Hence
$$\left(\frac{a}{d}\right)b<a+b\rightarrow \left(\frac{a}{d}-1\right)b<a$$
but since $a<b$, we conclude that $a=d$ and therefore $b=m$.
\end{proof}

By the previous lemmas we will always be able to select the decomposition $n=r_1\cdot r_2$ in such a way that $Ord_{r_1}(2)=m$. So, from now on we fixed the information set for $\C=R(1,m)$
\begin{equation}\label{conjinf}
 	I=\{0,\alpha^i\mid i\in \varphi^{-1}(\Gamma)\},
\end{equation}
where
$$\Gamma=\{(i_1,i_2)\in \Z_{r_1}\times\Z_{r_2}\mid 0\leq i_1\leq m, 0\leq i_2<1\}. $$

\vspace{.5cm}

Now, let $T_\alpha\in S_{G^*}$ be defined by 
 \begin{equation}\label{talfa}
T_\alpha\left(b X^0+\sum_{i=0}^{n-1} a_i X^{\alpha^i}\right)=b X^0+\sum_{i=0}^{n-1} a_i X^{\alpha^{i+1}}\end{equation}
Then $T_\alpha \in \PAut(\C^*)\subseteq\PAut(\C)$ because $\C$ is an extended cyclic code. Moreover, since $S_{G^*}\simeq S_{\Z_{r_1}\times\Z_{r_2}}$ (the group of automorphisms of $\Z_{r_1}\times\Z_{r_2}$) we have an isomorphism (induced by $\varphi$)
$$\phi:\langle T_\alpha \rangle\rightarrow \langle T_1,T_2\rangle $$
where $T_1(x,y)=(\overline{x+1},y)$ for any $(x,y)\in\Z_{r_1}\times\Z_{r_2}$, and $T_2$ is defined analogously.

We are proving that $\langle T_\alpha\rangle$ is a $s$-PD-like set for $\C^*$ with respect to $I$ for some natural number $s$. First we need the following technical lemma proved in \cite{BS2}. We denote by $[\cdot]_r$ the remainder modulo $r$.

\begin{lemma}\label{juntarpuntos}
 Let $r, x_1,\dots,x_h\in \N$ where $0\leq x_1<x_2<\dots<x_h<r$. Then there exists $\mu\in\N$ such that 
$$\left\lceil \frac{r}{h}\right\rceil-1\leq [x_i+\mu]_r<r$$
for all $i=1,\dots,h$ and $[x_j+\mu]_r=r-1$ for some $j\in \{1,\dots,h\}$.
\end{lemma}

\begin{theorem}\label{sPDlike}
Let $\alpha$ be a fixed $n$-th primitive root of unity and $\varphi: \Z_n\longrightarrow \Z_{r_1}\times\Z_{r_2}$ a fixed isomorphism. Let $T_\alpha\in S_{G^*}$ be the automorphism defined in (\ref{talfa}) and $I'=I\setminus\{0\}$, where $I$ is defined in (\ref{conjinf}). Then, the group generated by $T_\alpha$ is a $s$-PD-like set for $\C^*$ with respect to $I'$ where
$$s=(\lambda_0+1) r_2-1$$
and $\displaystyle \lambda_0=\max\left\{\lambda\mid m<\left\lceil \frac{r_1}{\lambda}\right\rceil\right\}$.
\end{theorem}

\begin{proof}
  Let $B\subseteq \{\alpha^i\}_{i=0}^{n-1}$, with $\left|B \right|=s$, and consider $B'=\{\varphi(i):\alpha^i \in B\}\subseteq \Z_{r_1}\times\Z_{r_2}.$ 
	
	Then, if we prove that there exists $\tau\in \langle T_1, T_2\rangle$ such that $\tau(B')\cap \Gamma=\emptyset$ then we will be done.
	
	Let us consider $\pi_i(B')$ the projection onto $\Z_{r_i}$ for $i=1,2$. For each $j\in \pi_2(B')$ we define $B_j=\{b\in B'\mid \pi_2(b)=j\}$. Then $B'=\bigcup\limits_{j\in \pi_2(B')} B_j$ and it is a disjoint union.
	
	First, suppose that $\pi_2(B')\subsetneq \Z_{r_2}$. Then we have that there exists $\delta\in \N$ such that $\pi_2\left(T_2^\delta(B')\right)\subseteq \{1,\dots,r_2-1\}$ so $T_2^\delta(B')\cap \Gamma=\emptyset$ and we are done.
	
	Now, suppose that $\pi_2(B')= \Z_{r_2}$. We claim that there must exist $j_0\in \pi_2(B')$ such that $\left|B_{j_0}\right|\leq \lambda_0$. Otherwise 
	 $$(\lambda_0+1) r_2-1=s=\left|B'\right|=\sum\limits_{j\in\pi_2(B')}\left|B_j\right|\geq (\lambda_0+1)\cdot r_2,$$
a contradiction.

  Then, let $\delta\in \N$ be such that $\pi_2\left(T_2^\delta(B_{j_0})\right)=0$ (and $\pi_2\left(T_2^\delta(B_{j})\right)\neq 0$ for any $j\neq j_0, j\in \Z_{r_2}$).
	
	By Lemma \ref{juntarpuntos}, applied to $r=r_1$ and the elements of $\pi_1\left(B_{j_0}\right)$ we have that there exists $\mu$ such that $\pi_1\left(T_1^\mu(B_{j_0})\right)\subseteq \{m,\dots, r_1-1\}$. So we conclude that
	$$T_1^\mu T_2^\delta(B')\cap\Gamma=\emptyset$$
	
	which finishes the proof.
\end{proof}

Then, we establish our final result.

\begin{theorem}\label{MainResult}
 Let $R(1,m)$ be the first-order Reed-Muller code of length $2^m$, ($m\in \N, m>2$). Let $n=2^m-1=r_1\cdot r_2$ with $gcd(r_1,r_2)=1, r_1,r_2>1$ and $Ord_{r_1}(2)=m$. Then we can correct up to $s$ errors by using {\scshape Algorithm II} where
 $$s=(\lambda_0+1)\cdot r_2-1$$
and $\lambda_0=max\{\lambda\mid m<\left\lceil \frac{r_1}{\lambda}\right\rceil\}$.
\end{theorem}

\begin{proof}
 By Lemma~\ref{descomposicion2} we can take the decomposition $n=r_1\cdot r_2$ satisfying the imposed conditions. 

Then we fix the information set $I$ given in (\ref{conjinf}). From Theorem~\ref{sPDlike} the group $\langle T_\alpha\rangle\subseteq \PAut(R^*(1,m))\subseteq\PAut(R(1,m))$ is a $s$-PD-like set for $R^*(1,m)$ with respect to $I$.

Therefore, by Theorem~\ref{PDnew} we can use {\scshape Algorithm II} successfully.
\end{proof}

\section{Examples}\label{Examples}

In this final section we exhibit Table \ref{tabla} that includes the first values of $m$ that satisfy conditions (\ref{condFirstOrder}). For them we apply Theorem \ref{MainResult} to give the number of errors we can correct by using {\scshape Algorithm II}. 

Let us observe that, for some values of $m$, there exist more than one valid decomposition. It is an open problem to find conditions on those decompositions that determine ``a priori'' the optimal one in order to correct the maximum number of errors. For such values of $m$ we only show the best decomposition.

On the other hand, we add the columns that corresponds to the results obtained for first-order Reed-Muller codes, $R(1,m)$, in the references we mentioned in the introduction. Specifically:

\begin{itemize}
	\item \textbf{Column A:} In \cite{KMM2} the authors mention that the group of translations of $\F_2^m$ is an $s$-PD-set where 
	$$s=\min\left\{\left\lfloor\frac{2^m-1}{1+m} \right\rfloor, 2^{m-2}-1\right\}$$
	We include the value of $s$.
	\item \textbf{Columns B1,B2:} In \cite{KMM3} the authors explain that there is an $(M-1)$-PD-set provided there exists a $(m,M,3)$ binary code. They also show that for any $1\leq s\leq \left\lfloor\frac{2^m}{m+1} \right\rfloor-1$ there exist $s$-PD-sets. We include the values $(M-1)$ and $s$ respectively.
\end{itemize}

\begin{table}[h]
\caption{Number of corrected errors}
\label{tabla}
\centering

\begin{tabular}{|c|c|c||c|c||c|c|c|c|}
 \hline &&&&&&&&\\
 $m$ & $r_1$ & $r_2$ &$l$&$t$&  $A$ & $B1$ & $B2$& {\scshape Algorithm II}\\  &&&&&&&&\\ \hline
4&5&3&16&8&   3&1&2&5\\ \hline
6&9&7&64&32&  5&3&8&13\\ \hline
$8^*$&17&15&256&128&   28&16&27&44\\ \hline
9&73&7&512&256&   51&32&50&62\\ \hline
$10^*$&11&93&1024&512&  93&64&92&185\\ \hline
$11^*$&23&89&2048&1024&   170&128&169&266\\ \hline
$12^*$&13&315&4096&2048&   315&256&314&629\\ \hline
$14^*$&43&381&16384&8192&   1092&1024&1091&1523\\ \hline
$15^*$&151&217&32768&16384&   2047&2048&2047&2386\\ \hline
$16^*$&257&255&65536&32768&   3855&2048&3854&4334\\ \hline
\end{tabular}
\end{table}

The columns $l$ and $t$ include the length and the error-correction capability of the Reed-Muller code respectively. Finally, we write the $(^*)$ for the cases where there exist more than one valid decomposition. 

\begin{remark}
  Observe that according to the Gordon-Schönheim bound (mentioned in the introduction), the minimal size of a $s$-PD-set is precisely $s+1$, that is, 1 plus the value contained in column $A$.
\end{remark}

\end{document}